\documentclass[american,aps,pra,reprint,superscriptaddress,showpacs]{revtex4-1}
\usepackage[T1]{fontenc}
\usepackage[latin9]{inputenc}
\usepackage{color}
\usepackage{amsthm}
\usepackage{amsmath, amssymb}
\usepackage{color}
\usepackage{bm}
\usepackage{bbm}
\usepackage{empheq}
\usepackage{pxfonts,txfonts}
\usepackage{tgpagella}
\usepackage[T1]{fontenc}
\usepackage{graphicx}
\usepackage{txfonts}
\usepackage{braket}
\usepackage{comment}
\usepackage[unicode=true,pdfusetitle,
 bookmarks=true,bookmarksnumbered=false,bookmarksopen=false,
 breaklinks=false,pdfborder={0 0 0},backref=false,colorlinks=true]
 {hyperref}
\hypersetup{
 linkcolor=urlcolor,citecolor=urlcolor,urlcolor=urlcolor}

\makeatletter
\@ifundefined{textcolor}{}
{%
 \definecolor{BLACK}{gray}{0}
 \definecolor{WHITE}{gray}{1}
 \definecolor{RED}{rgb}{1,0,0}
 \definecolor{GREEN}{rgb}{0,1,0}
 \definecolor{BLUE}{rgb}{0,0,1}
 \definecolor{CYAN}{cmyk}{1,0,0,0}
 \definecolor{MAGENTA}{cmyk}{0,1,0,0}
 \definecolor{YELLOW}{cmyk}{0,0,1,0}
 }
\theoremstyle{plain}
\newtheorem{thm}{\protect\theoremname}
\theoremstyle{plain}

\ifx\proof\undefined
\newenvironment{proof}[1][\protect\proofname]{\par
\normalfont\topsep6\p@\@plus6\p@\relax
\trivlist
\itemindent\parindent
\item[\hskip\labelsep
\scshape
#1]\ignorespaces
}{%
\endtrivlist\@endpefalse
}
\providecommand{\proofname}{Proof}
\fi
\theoremstyle{plain}
\newtheorem{lem}[thm]{\protect\lemmaname}
\theoremstyle{plain}
\newtheorem{cor}[thm]{\protect\corrolaryname}
\theoremstyle{plain}

\definecolor{urlcolor}{rgb}{0,0,0.7}
\newcommand{\id}{\mathbbm{1}}
\newcommand{\CC}{\mathbbm{C}}

\renewcommand{\vec}[1]{\text{\boldmath$#1$}}

\makeatother

\providecommand{\lemmaname}{Lemma}
\providecommand{\propositionname}{Proposition}
\providecommand{\theoremname}{Theorem}
\providecommand{\corrolaryname}{Corollary}
\providecommand{\conjecturename}{Conjecture}

\begin{document}

\title{Linear maps as sufficient criteria for entanglement depth and compatibility in many-body systems}

\author{Maciej Lewenstein}
\email{maciej.lewenstein@icfo.es}
\affiliation{ICFO -- Institut de Ci\`encies Fot\`oniques, The Barcelona Institute
of Science and Technology, 08860 Castelldefels, Spain}
\affiliation{ICREA -- Instituci\'o Catalana de Recerca i Estudis Avan\c cats, Lluis
Companys 23, 08010 Barcelona, Spain}

\author{Guillem M\"uller-Rigat}
\affiliation{ICFO -- Institut de Ci\`encies Fot\`oniques, The Barcelona Institute
of Science and Technology, 08860 Castelldefels, Spain}

\author{Jordi Tura}
\affiliation{Instituut-Lorentz, Universiteit Leiden, P.O. Box 9506, 2300 RA Leiden, The Netherlands}

\author{Anna Sanpera}
\affiliation{F\'isica Te\'orica: Informaci\'o i Fen\`omens Qu\`antics.  Departament de F\'isica, Universitat Aut\`onoma de Barcelona, E-08193 Bellaterra, Spain}
\affiliation{ICREA -- Instituci\'o Catalana de Recerca i Estudis Avan\c cats, Lluis
Companys 23, 08010 Barcelona, Spain}

\begin{abstract}
Physical transformations are described by linear maps that are completely positive and trace preserving (CPTP). However, maps that are positive (P) but not completely positive (CP) are instrumental to derive separability/entanglement criteria. Moreover, the properties of such maps can be linked to entanglement properties of the states they detect. Here, we extend the results presented in [Phys. Rev A {\bf 93}, 042335 (2016)], where sufficient separability criteria for bipartite systems were derived. In particular, we analyze the entanglement depth of an $N$-qubit system by proposing linear maps that, when applied to any state, result in a bi-separable state for the $1:(N-1)$ partitions, i.e., $(N-1)$-entanglement depth. Furthermore, we derive criteria to detect arbitrary  $(N-n)$-entanglement depth tailored to states in close vicinity of the completely depolarized state (the normalized identity matrix). We also provide separability (or $1$- entanglement depth) conditions in the symmetric sector, including for diagonal states. Finally, we suggest how similar map techniques can be used to derive sufficient conditions for a set of expectation values to be compatible with separable states or local-hidden-variable theories. We dedicate this paper to the memory of the late Andrzej Kossakowski, our spiritual  and intellectual mentor in the field of linear maps.

\end{abstract}

\pacs{}

\maketitle

\section{Introduction}

Linear maps are the fundamental tools to describe transformations in quantum systems. In fact, any physical transformation corresponds to a completely positive trace preserving operation (CPTP), that is, a linear map, $\Lambda$,  that is positive (P), $\Lambda(\rho)\geq 0$,  on all quantum states ($\rho\geq 0$),  and whose extension to a larger tensorial space leaves all states in the composite space positive as well $(\id_{A} \otimes \Lambda_{B})\rho_{AB}\geq 0$. Such description applies for closed and open quantum systems irrespectively from the fact that the quantum system of interest is single or multipartite. However, maps that are positive (P) but not completely positive (CP) are instrumental to derive separability/entanglement criteria, that is, the absence/presence of shared quantum correlations between the constituents of multipartite systems. Moreover,  the properties of such maps can be linked to entanglement properties of the states they detect. Entanglement, the cornerstone feature of quantum systems, has no classical analogue, and is a fundamental resource for quantum information tasks. 

Entanglement has been mostly studied and addressed and exploited in the bipartite scenario. But even for this simplest case, checking if a given composite system is separable (entangled) is an NP-hard problem\cite{Gurvits.JCSS.2003}. Despite the fact that in low dimensions one can use semi-definite programming to numerically determine the presence of entanglement \cite{Doherty+2.PRL.2002}, such method becomes unfeasible when the dimensions of the composite Hilbert space grows, when the number of involved parties increases, or when the aim,  rather than checking the entanglement of a specific quantum state, is to is characterize the entanglement properties of families of states.  For all these cases the common approaches include using separability criteria or constructing entanglement witnesses.

 To date, most of the known separability criteria are {\it necessary} but not {\it sufficient}, 
starting with the
prominent positive partial transpose criterion, proposed by
Peres \cite{Peres.PRL.1996}. In fact transposition $T$ is a
paradigmatic example of a {\it positive map}, i.e., a map that maps
states onto states. Transposition is not, however, completely
positive, i.e., its extensions by identity map to larger tensor
spaces do not correspond to positive maps. Therefore, this condition becomes a necessary criterion of separability. Such map however, has proven to be sufficient condition of separability only for
two-qubits and one qubit-qutrit by Horodeckis
\cite{3Horodecki.PLA.1996}. For all bipartite systems (i.e. Alice-Bob),  it has also been proven in (cf.\cite{4Horodeckis.RMP.2009}), that any positive map acting exclusively
on a space of Alice's (Bob's) operators, transforms a separable state
into a positive definite state. Conversely, a state is separable
if it remains positive definite under the action of all positive maps
on Alice's (Bob's) side. Well known positive maps that have been
proposed and applied to determine separability are the
celebrated reduction map \cite{M+PHorodecki.PRA.1999,Cerf+2.PRA.1999},
the Breuer-Hall map
\cite{Breuer.PRL.2006,Hall.JPA.2006}, or the Choi map \cite{Choi.LAA.1975}.

Here, we follow the approach used to derive sufficient criteria of separability for bipartite systems (necessary criteria of entanglement)  previously introduced by some of us in \cite{LAC16}.  Now,  our first aim is to derive sufficient criteria to determine the  entanglement depth of an ensemble of $N$-qubit systems. The entanglement depth $n$ of an ensemble of $N$-parties expresses the fact that the quantum state of the ensemble necessary involves $n\leq N$ genuine multipartite entanglement.  This characterization of multipartite entanglement becomes extremely relevant for experiments involving a large number of parties as it happens e.g., with cold gases. It is therefore crucial for a meaningful comprehension of experimental results to have at our disposal operational entanglement classification that can be used in experiments.

The main point of our approach is to use families of
(positive) maps, $\Lambda_{\vec p}$,  acting on global states of an $N$-qubit (or qudit) multiparty system. These maps have the property that when applied to a state (a state of a certain
entanglement class) they produce  a "weaker" entanglement state (a state of
another certain entanglement class; separable, bi-separable, $n$-entanglement depth...). In particular, first we derive maps whose action in the states of $(\mathbb C^2)^{\otimes N}$ is to make them bi-separable with respect the one partition, e.g., 1:(N-1). In other words, maps all physical states onto the set of states with entanglement depth $n=N-1$. The inverted map, $\Lambda^{-1}_{\vec p}$, (if it exists) becomes then a sufficient criteria for bi-separability.  Such a method can be extended to derive maps that ensure (are sufficient) to reveal which type of entanglement  (entanglement depth)  the multipartite system has.  Further we exploit such idea to derive sufficient conditions for compatibility with separable states or local-hidden variable theories given a set of accessible expectation values.

The paper is organized as follows. In Section \ref{Preliminaries}
we present our notation, definitions, and preliminary facts. In
Section \ref{General theorem} we explain the general theorem. Section \ref{Main results} we present our main results for $N$-qubits with several maps leading to sufficient $N-1$ entanglement depth criteria and beyond. We discuss as well the particular case symmetric subspaces. In Section \ref{Other proposals}, we focus on deriving sufficient conditions for compatibility with separable states or local-hidden variable theories for a set of accessible expectation values. We close the paper with the conclusions and outlook in Section~\ref{Conclusions and outlook}.

\section{\label{Preliminaries} Preliminaries}

We consider  $N$-qubit system (Alice, Bob, Charlie, ... ) corresponding to a composite Hilbert space ${\cal H}={\cal H}_A\otimes{\cal H}_B\otimes\ldots= \CC^2\otimes \CC^2\otimes\ldots \CC^2=(\CC^2)^{\otimes N}$ of dimension $2^{N}$. The space of all bounded
operators acting on it is denoted by ${\cal B}({\cal H}_A\otimes{\cal H}_B\otimes\ldots)$,
while the set of states, that is, operators from ${\cal B}({\cal H}_A\otimes{\cal H}_B\otimes\ldots)$
obeying $\rho\ge 0$, $\rho =\rho^\dag$
and ${\rm Tr}(\rho)=1$, as $R$. Note that $R$ is convex and compact. However, the
normalization condition of the states and, consequently, the
compactness of $R$ will only be assumed if necessary. Partially transposed operators are denoted $\rho^T_{i}$, for $i=1,\cdots, N$. In the following, we will
consider also various classes of states such as:

\begin{itemize}

\item Separable states $\Sigma$, also denoted as $\Sigma_1$, i.e.,
states that admit the following decomposition
\begin{equation}
\sigma=\sum_{k=1}^K p_k |\Psi_k\rangle_1 \langle \Psi_k|_1,
\label{sigma}
\end{equation}
where $|\Psi_k\rangle_1= |e_{k_1}\rangle \otimes |e_{k_2}\rangle\otimes\ldots=\bigotimes_{i=1}^{N} |e_{k_{i}}\rangle$ are $N$ tensor products \cite{Peres.S.1995}.

\item States with entanglement depth  $n\le N$, $\Sigma_n$, i.e., the states that admit
the decomposition
\begin{equation}
\sigma=\sum_{k=1}^K p_k |\Psi_k\rangle_n \langle \Psi_k|_n,
\label{sigma}
\end{equation}
where $|\Psi_k\rangle_n= |e_{k_{i\cdots j}}\rangle \otimes |e_{k_{i',\cdots,j'}}\rangle \otimes \ldots$ 
are product vectors that include genuinely entangled states of  $i,\cdots, j\le n$ parties maximally, and at least one of the pure states must include an $n$-partite genuinely entangled factor.
\cite{TuraAloy1,TuraAloy2,TuraAloy3}.

\item PPT-operators $W_{PPT}$, are operators that have all possible partial transposes positive, i.e.,  $W^{T_i}\ge$ for all $i=1,\cdots, N$ but also  $W^{T_{N/i,..j}}\ge 0$, etc. in particular, we will talk about PPT-states, $R_{PPT}$, i.e.,
$\rho\in R$ and $\rho\in W_{PPT}$, i.e., $\rho^{T_X}\ge 0$, for any partition  of the set of qubits into $X$ and the rest $N/X$.

\end{itemize}

Further, we consider various classes of linear maps \[
\Lambda\colon{\cal B}({\cal H}_A\otimes{\cal H}_B\otimes\ldots) \to {\cal
B}({\cal H}_A\otimes{\cal H}_B\otimes\ldots)\] that preserve hermiticity. In particular,
we pay particular attention to family of maps parametrized by a multi-index $\bf p$  that map a certain subset $S\subset {\cal B}({\cal H}_A\otimes{\cal H}_B\otimes\ldots)$ to another
certain subset of $S'\subset {\cal B}({\cal H}_A\otimes{\cal
H}_B\otimes\ldots)$

\begin{equation}
\Lambda_{{\bf p}}\colon S \to S'
\end{equation}

such that for every $w \in S$, then $\Lambda_{\bf p}(w)\in S'$.  We assume that those maps are invertible for almost  all values of ${\bf p}$.

\section{\label{General theorem} General theorem}

First, we review the starting point of \cite{LAC16}, which is build upon the simple observation: 

\begin{thm}
\label{thm:general} Let  $S, S'$ be convex and compact subsets of
${\cal B}({\cal H}_A\otimes{\cal H}_B\otimes\ldots)$, and let $\Lambda_{\bf p}\colon S\to S'$
be a family of maps, invertible for almost all $\bf p$. By 
${\cal P}_{SS'}$  we denote the subset of the parameters set ${\bf p}$. The maps have the property that for  every $w\in S$, 
$\Lambda_{\bf p}(w)\in S'$ provided ${\bf p}\in {\cal P}_{SS'}$. Then if $\Lambda^{-1}_{\bf p}(\sigma) \in S$ $\Rightarrow$
$\sigma \in S'$.

\end{thm}
\begin{proof}
Note that $\Lambda_{\bf p}\left[\Lambda^{-1}_{\bf p}(\sigma)\right]=\sigma$.
\end{proof}

In other words, the theorem provides sufficient criteria for $S'$. In practice, the choice of $\Lambda_\mathbf{p}$ is such that:

\begin{itemize}
	\item[i)]  we can easily
	check that $\Lambda_{\bf p}^{-1}(S')\subset S$,
	\item[ii)] we can prove the
	assumption that $\Lambda_{\bf p}(S)\subset S'$.
\end{itemize}

 The difficulty on the derivation of such criteria is hidden in demonstrating condition ii) and in technical difficulties of inverting $\Lambda_{\bf p}$.

In this work, we generalize \cite{LAC16} to the multipartite case by initially considering $S$ the set of all semidefinite-positive elements of $\mathcal{B}((\mathbb{C}^2)^{\otimes N})$. Consequently, the sufficient condition will be written as $\Lambda^{-1}_\mathbf{p}(\sigma)\geq 0 \Longrightarrow \sigma\in S' $. The concrete sets $S'$ that we will examine are related to the entanglement properties of $N$-party quantum states and are to be specified in the next section for each example.  Our approach bears similitudes with the works of Barnum and Gurvits \cite{Gurvits+Barnum.PRA.2002} and Szarek {\it et al.} \cite{Szarek.PRA.2005,Aubrun+Szarek.PRA.2006}

\section{\label{Main results} Main results}

Here we present the main results on the construction of new sufficient criteria for $N$-qubit systems from the inversion of linear maps. We will consider families of such transformations inspired by the reduction and the Breuer-Hall maps. 

\subsection{Reduction-inspired maps}

To start with, we consider the simple reduction-type family of maps defined as $\Lambda_\alpha(\rho) = \mathrm{Tr}(\rho)\id +\alpha\rho $.

\subsubsection{Sufficient criteria for $N-1$ - entanglement depth}

Now we fix $S'\subset S$ the set of bi-separable states in the partitions $1:(N-1)$, (one qubit subsystem versus the others $N-1$). In other words, $S'$ is the set of  states with entanglement depth $(N-1)$, so the set does not contain genuine $N$-qubit entanglement. 
The task is to find the range of $\alpha$ in which for all states $\rho\geq 0$ we have that $\Lambda_\alpha(\rho)$ is bi-separable for any partition $1:(N-1)$. Once this is done, $\Lambda_{\alpha}^{-1}(\rho)\geq 0$ would be a sufficient $(N-1)$-entanglement depth criterion.  Bi-separability of $\Lambda_{\alpha}(\rho)$ for $N$ qubits follows from the following known result.

\begin{thm}
\label{thm:barnum} Compare \cite{Vidal.PRA.1999,Gurvits+Barnum.PRA.2002}. Let $\Lambda_\alpha(\rho)=
{\rm Tr}(\rho)\id
+ \alpha \rho $ be the family of maps, and
$-1\le\alpha \le 2$. Then $\rho\ge 0$ $\Rightarrow$ $\Lambda_\alpha(\rho)=:\sigma \in \Sigma_{N-1}$,
 i.e. $\sigma$ is bi-separable in any $1:(N-1)$ partition (i.e. has entanglement depth $(N-1)$)
\end{thm}

\begin{proof}
The proof follows \cite{LAC16}. It is enough to prove the theorem for pure states, $\rho=|\Psi\rangle\langle
\Psi|$. For $N$ qubits, $|\Psi\rangle$ has maximally Schmidt rank 2, in the partition Alice:All Others (any other $1:(N-1)$). Without loosing generality, we can assume that $|\Psi\rangle
=\lambda_0 |0\rangle\otimes |0\rangle_{N-1} + \lambda_1|1\rangle\otimes
|1\rangle_{N-1}$. It is then enough to check the positiveness and separability
of $\Lambda(|\Psi\rangle\langle \Psi|)$ on a $2\otimes 2$ space
spanned by $|0\rangle$, $|1\rangle$, $|0\rangle_{N-1}$, $|1\rangle_{N-1}$ in both Alice's and All Others'
spaces. PPT provides then necessary and sufficient condition \cite{3Horodecki.PLA.1996}, and we easily get that indeed $-1\le
\alpha\le 2$. The condition $\alpha \ge -1$ actually follows
already from a simpler requirement of positiveness of
$\Lambda_{\alpha}(\rho)$, as it appears in the definition of the reduction
map \cite{M+PHorodecki.PRA.1999,Cerf+2.PRA.1999}.
\end{proof}

The inverse of the reduction map leads to the following condition: 

\begin{cor}
\label{cor:barnum1} {\bf (Sufficient $N-1$-entanglement depth 
criterion 1)} If
$\Lambda^{-1}_\alpha(\sigma):= [\sigma -{\rm
Tr}(\sigma)\id/(2^N+\alpha)]/\alpha \ge 0$, then $\sigma\in \Sigma_{N-1}$,
i.e. $\sigma$ is bi-separable in  any partition $1:(N-1)$.
\end{cor}

It is interesting to check whether simple criteria are strong enough to detect some states outside the "separable ball" around identity for $N$ qubits\cite{Szarek.PRA.2005,Aubrun+Szarek.PRA.2006,Gurvits+Barnum.PRA.2002} Of course, we do not expect that all states from the separable ball are detected. The reason for this weakness is that Barnum and Gurvits, and Szarek have estimated the radius of the separable ball using a stronger result, namely: if the operator norm of a Hermitian matrix $X$ is bounded by 1, then $\id+ X$ is separable. In contrast, our simple criteria depend on the spectrum, not on the norm, and are thereby independent. Nonetheless, we will see in the next sections that with other choices of $\Lambda_{\bf p}$'s, we could derive further stronger sufficient $(N-1)$-depth criteria.

\subsubsection{Sufficient criteria for arbitrary $N-n$-entanglement depth }

The previous result can be generalized to detect arbitrary $N-n$-entanglement depth around the maximally mixed state with the next Theorem, which is analogous to Theorem \ref{thm:barnum}: 

\begin{thm}
\label{thm:clarisse1} Let $\Lambda_\alpha(\rho)=
{\rm Tr}(\rho)\id + \alpha \rho $ be the family of maps, and
$-1\le\alpha \le 2$. Then $\rho\ge 0$ $\Rightarrow$
$\Lambda_\alpha(\rho)=\sigma\in \Sigma_{N-n}$, i.e. $\sigma$ has the
entanglement depth $N-n$. 

Similarly, if $\sigma \ge 0$, and $\Lambda_\alpha^{-1}(\sigma)=\rho\ge 0$, then $\sigma$ has the
entanglement depth $N-n$. 
\end{thm}

\begin{proof}
We discuss in detail the cases $n=2$; the generalizations to
arbitrary $n$ are  straightforward. As before, it  is
enough to prove for pure states, $\rho=|\Psi\rangle\langle \Psi|$.
For $n=2$, $|\Psi\rangle$ has maximally Schmidt rank 4, and,
without loosing generality, we can assume that $|\Psi\rangle
=\lambda_0 |0\rangle\otimes |0\rangle + \lambda_1|1\rangle\otimes
|1\rangle + \lambda_2|2\rangle\otimes |2\rangle +  \lambda_3|3\rangle\otimes |3\rangle$, with $\lambda_0\ge \lambda_1\ge \lambda_3\ge\lambda_4$, and $\sum_{i=0}^3 \lambda_i^2=1$. It is enough to
check then positiveness and separability of
$\Lambda(|\Psi\rangle\langle \Psi|)$ on a $4\otimes 4$ space
spanned by $|0\rangle$, $|1\rangle$, $|2\rangle$  and $|3\rangle$ in both
Alice's and Bob's spaces. This is achieved by observing first that
if $|\Psi\rangle$ is a normalized  state, then obviously the positivity of
$\Lambda_\alpha(|\Psi\rangle\langle \Psi|)$ requires $\alpha\ge
-1$. Direct inspection shows that for $\alpha=-1$, $\Lambda_\alpha(|\Psi\rangle\langle \Psi|)$  is a sum of four
positive matrices that effectively act on two-qubit Hilbert spaces and are given by
%
%
\begin{equation}
\sigma_{ij}=\ket{\varphi_{ij}}\!\bra{\varphi_{ij}}+\ket{ij}\!\bra{ij}+\ket{ji}\!\bra{ji},
\end{equation}
where $\ket{\varphi_{ij}}=\lambda_j\ket{ii}-\lambda_i\ket{jj}$ with $i<j=0,1,2,3$.

These matrices are PPT and thus separable.

For $\alpha\ge 0$, we decompose $\Lambda_\alpha(|\Psi\rangle\langle \Psi|)
=\sigma^{T_A} + D$, where $\sigma^{T_A}$ is separable by construction,
and $D$ is positive diagonal in the computational product basis,
{\it ergo} separable by direct inspection. We consider a family of
product vectors
$|p(\phi,\psi)\rangle=A(1,ae^{i\phi},be^{i\psi},ce^{i\theta})^{\otimes 2}$ and
the separable state
\[\sigma= \int d\phi/2\pi \int d\psi/2\pi \int d\theta/2\pi|p(\phi,\psi)\rangle \langle
p(\phi,\psi)|.\] The parameters can be chosen indeed in such a way
that $\Lambda_\alpha(|\Psi\rangle\langle \Psi|) -\sigma^{T_A}$ is
diagonal. To this aim we set $A^2=\alpha\lambda_0^2$,
$a^2=\lambda_1/\lambda_0$, $b^2=\lambda_2/\lambda_0$, and
$c^2=\lambda_3/\lambda_0$. Checking the
explicit conditions that the matrix $D$ is positive implies that
$1-A^2a^2=1-\alpha\lambda_0\lambda_1\ge 0$, i.e. $\alpha\le 2$
since $\lambda_0$ and $\lambda_1$ are the highest and the second-highest Schmidt coefficients, respectively.
The case $-1<\alpha<0$ follows from convexity.

\end{proof}

\subsubsection{Sufficient separability (or 1- entanglement depth) criteria for symmetric states}

In this subsection, we will start by considering maps acting on the bosonic space of $N$ qubits, ${\cal S}(\CC^{2}\otimes\ldots\otimes \CC^{2})$ by setting $S$ the set of $N$-qubit symmetric states. The dimension of such subspace is $N-1$ and the basis elements read: 

\begin{equation}
    \left\{\ket{S_k^N} ={N\choose k}^{-1/2} \sum_{\pi\in \mathfrak{S}_{N}} \pi(\ket{0}^{\otimes N} \otimes \ket{1}^{\otimes {N-k}})\right\}_{k=0}^N \ ,
    \label{Definition_Dicke}
\end{equation}
where the sum is over the distinct permutations. In the context of the reduction map, we denote $\id_\mathcal{S}$ as the identity in the symmetric subspace (i.e. the subspace spanned by the vectors of Eq. \ref{Definition_Dicke}).

In this section. we start by deriving sufficient separability criteria by taking $S'$ the subset of separable symmetric states of two or three qubits. 

For $N=2$ we obtain:

\begin{thm}
\label{thm:symm} Let $\Lambda_\alpha(\rho_\mathcal{S})=
{\rm Tr}(\rho_\mathcal{S})\id_\mathcal{S}+ \alpha \rho_\mathcal{S} $ be the family of maps  acting on the states of two qubits in the symmetric space, and
$-3/4\le\alpha \le 1$. Then $\rho_\mathcal{S}\ge 0$ $\Rightarrow$
$\Lambda_\alpha(\rho_\mathcal{S})=\sigma_\mathcal{S}\in \Sigma_{1}$, i.e. $\sigma_\mathcal{S}$ is separable.

Similarly, if $\sigma_\mathcal{S} \ge 0$, and $\Lambda_\alpha^{-1}(\sigma_\mathcal{S})=\rho_\mathcal{S}\ge 0$, then $\sigma_\mathcal{S}$ has the
entanglement depth $1$. 
\end{thm}

\begin{proof}

The proof is straightforward. We consider a pure state, $|\Psi\rangle =\lambda_0 |0\rangle\otimes
|0\rangle + \lambda_1|1\rangle\otimes |1\rangle$, apply the map to it and check the positivness of the result and its separability using the PPT criterion. 
\end{proof}

Similar, but a little weaker result can be obtained for the maps acting on the states of $N=3$ qubits in the symmetric space. In this subspace, PPT is necessary but also sufficient for separability \cite{Eckert2002}. 

\begin{thm}
\label{thm:symm1} Let $\Lambda_\alpha(\rho_\mathcal{S})=
{\rm Tr}(\rho_\mathcal{S})\id_\mathcal{S}+ \alpha \rho_\mathcal{S} $ be the family of maps  acting on the states of three qubits in the symmetric space, and
$-1/3=-\lambda_{\rm min}\le\alpha \le 2 \lambda_{\rm min}=2/3$, where $\lambda_{\rm min}>0$ is the minimal eigenvalue of $(\id_{\cal S})^{T_A}$. Then $\rho_\mathcal{S}\ge 0$ $\Rightarrow$
$\Lambda_\alpha(\rho_\mathcal{S})=\sigma_\mathcal{S}\in \Sigma_{1}$, i.e. $\sigma_\mathcal{S}$ is separable.

Similarly, if $\sigma_\mathcal{S} \ge 0$, and $\Lambda_\alpha^{-1}(\sigma_\mathcal{S})=\rho_\mathcal{S}\ge 0$, then $\sigma_\mathcal{S}$ has the
entanglement depth $1$. 
\end{thm}

\begin{proof}
Obviously, the positiveness of the map requires $\alpha\ge -1$. Now we observe that $(\id_{\cal S})^{T_A}$ is a positive definite operator acting on $\CC^2\otimes{\cal S}(\CC^2\otimes\CC^2)$, with minimal eigenvalue $\lambda_{\rm min}$. The partial transpose of any projector on a pure state has  eigenvalues between $[-1/2,1]$\cite{Swapan_neg}, so for positive $\alpha$, we must have $\lambda_{\rm min}-\alpha/2\ge 0$, and for negative $\alpha$, $\lambda_{\rm min}-|\alpha|\ge 0$. The eigenvalues of $(\id_{\cal S})^{T_A}$ can be obtained analytically and are: $(1/3,1/3,1/3, 1/3, 4/3,4/3)$, so that $\lambda_{\rm min}=1/3$. 
\end{proof}

The latter result can be generalized to $N$ qubits, where $\lambda_{\rm min}=1/N$. There, we cannot conclude PPT states are separable: starting from $N=4$, there exist PPT entangled symmetric states \cite{Tura-sym}. \\

\textbf{Diagonal states.--} However, for diagonal qubit symmetric states for arbitrary $N$,

\begin{equation}
    \rho_{DS} = \sum_{k=0}^N p_k \ket{S_k^N}\bra{S_k^N} \ ,
\end{equation}

PPT for the largest bipartitions $\left\lfloor {\frac{N}{2}} \right\rfloor:N-\left\lfloor {\frac{N}{2}} \right\rfloor$ is necessary and sufficient for separability \cite{Ruben_Sanpera}. The next result gives an efficient way to compute the partial transpose for such states:

\begin{lem}\label{Hankel} \cite{Yu, Ruben_Sanpera}
The state $\rho_{DS}$ has positive partial transpose with respect the bipartition $\left\lfloor {\frac{N}{2}} \right\rfloor:N-\left\lfloor {\frac{N}{2}} \right\rfloor$ iff

\begin{equation}
M_{l}(\{p_k\}_{k=0}^{N})= \left[
\begin{matrix}
p_{0+l}/{N\choose 0+l}&p_{1+l}/{N\choose 1+l} &  \dots& p_{n_l+l}/{N\choose n_l+l} \\
p_{1+l}/{N\choose 1+l}&p_{2+l}/{N\choose 2+l}& \dots  &p_{n_l+1+l}/{N\choose n_l+1+l} \\
\vdots & \vdots &  &  \vdots \\
p_{n_l+ l}/{N\choose n_l+l} & p_{n_l +1 + l}/{N\choose n_l+1+l}  & \cdots & p_{2n_l + l}/{N\choose 2n_l+l}
\end{matrix}\right] \geq 0,\\
\end{equation}
where $l\in\{0,1\}$ and $n_{l=0} = \left\lfloor {\frac{N}{2}} \right\rfloor, n_{l=1} = \left\lfloor {\frac{N-1}{2}} \right\rfloor $. 
\end{lem}

This allows us to formulate the basic result on diagonal $N$-qubit symmetric states: 

\begin{thm}
\label{thm:symm4} Let $\Lambda_\alpha(\rho_{DS})=\sigma_{DS}=
{\rm Tr}(\rho)\id_\mathcal{S} + \alpha \rho_{DS}$ be the family of maps  acting on diagonal qubit symmetric states. Let $\alpha$ be such that the two associated matrices $M'_0$ and $M'_1$ are positive semidefinite, 
where  $M'_0=M_0(\{1+\alpha p_k\}_{k=0}^{N})$, $M'_1= M_1(\{1+\alpha p_k\}_{k=0}^{N})$.

Similarly, if $\sigma_{DS} \ge 0$, and $\Lambda_\alpha^{-1}(\sigma_{DS})=\rho_{DS}\ge 0$, then $\sigma_{DS}$ is separable. 
\end{thm}
\begin{proof}
The map transforms $p_k \to p_k'=1+\alpha p_k$. Note that $p_k'$ are not normalized, but this does not affect positivity condition.
Thus, the relevant Hankel matrices for $\sigma_{DS}$ have indeed the form presented in Lemma \ref{Hankel}.
\end{proof}

 Via a convex roof argument, it is sufficient to consider pure states. Hence, to compute the bounds of $\alpha$ for diagonal symmetric states it is sufficient to check $\Lambda_\alpha(\ket{D^N_k}\bra{D^N_k})$ for each basis element $k=0,1,..,N$  and take the intersection of ranges $\alpha_{{\rm min, max}_k}$. In order to compute the bound of $\alpha$ we know that for $\alpha =0$ the state is PPT so we just need to solve $\det M'_0 = 0,\det M'_1 = 0 $ with $M_{0},M_{1} $ corresponding to each basis element to find at least one vector in the kernel. By finding the closest roots of $\det M'_0 = 0,\det M'_1 = 0 $ with $M_{0},M_{1} $ from both sides, we obtain
\begin{table}[h]
\begin{center}
\begin{tabular}{||c c c c||} 
 \hline
 $N = 2$ & $N=3$ & $N=4$ & $N=5$ \\ [0.5ex] 
 \hline\hline
 $(-0.750,1.000)$ & $(-0.667,0.732)$ & $(-0.326,0.440)$  & $(-0.231,0.268)$ \\ 
 \hline
\end{tabular}
\end{center}
\caption{Bounds on $\alpha$, $(\alpha_{\rm{min}}, \alpha_{\rm{max})}$ for different number of parties $N$. The bounds $(\alpha_{\min},\alpha_{\max})$ are algebraic numbers and they take the following closed expressions: for $N=2: [-3/4,1]$, for $N=3: [-2/3, \sqrt{3}-1]$, for $N=4: [(3\sqrt{6}-8)/2, 5(\sqrt{21}-3)/18]$, and for $N=5$ the lower bound is given by the smallest root of the polynomial $18\alpha^3 + 54\alpha^2-378\alpha -90$ and the upper bound is $(\sqrt{114}-8)/10.$}
\end{table}

This result can be extended to qudits. In such case, the symmetric basis is generalized as: 

\begin{equation}
    \left\{ \ket{S_\mathbf{k}^N} ={N\choose \mathbf{k}}^{-1/2} \sum_{\pi \in \mathfrak{S}_{N}} \pi\left(\bigotimes_{a=0}^{d-1}\ket{a}^{\otimes k_a}\right)\right\}_{\mathbf{k}\vdash N} \ ,
\end{equation}
where now $\mathbf{k} = (k_0,k_1,..,k_{d-1})$ is a partition of $N$. Presently we consider states of the form introduced in Ref. \cite{Rut_Mar}:
\begin{equation}
    \rho_{DD} = \sum_{k=0}^{N(d-1)}p_k\ket{D_k^{N}}\bra{D_k^{N}} \ ,
    \label{dicke_diagonal}
\end{equation}
where $\ket{D_k^N}$ are the natural generalization of Dicke states for higher local dimension and read as coherent superposition of product vectors with the same number of excitations. In the symmetric basis, 

\begin{equation}
    \left\{ \ket{D_k^N} = \mathcal{N}_k^{-1/2}  \sum_{ \mathbf{d}\cdot \mathbf{k} = k}{N\choose \mathbf{k}}^{1/2}\ket{S_\mathbf{k}^N}\right\}_{k=0}^{N(d-1)} \ ,
 \label{Dicke_d_ basis}
\end{equation}
where $\mathcal{N}_k^{1/2}$ is the normalization factor  (without simple formula) and $\mathbf{d} = (0,1,..,d-1)$. We denote by $\id_\mathcal{D}$ the identity (projector) on the basis of Eq. \eqref{Dicke_d_ basis}.  

Similarly, for $N$ even PPT with respect the largest bipartition $\left\lfloor {\frac{N}{2}} \right\rfloor:\left(N-\left\lfloor {\frac{N}{2}} \right\rfloor\right)$ is a necessary and sufficient separability criteria for states of the form $\rho_{DD}$ \cite{Rut_Mar}. In complete analogy with the qubit case, we can pose the following theorems:    

\begin{lem}\label{Hankeld}
\cite{Rut_Mar}
The state $\rho_{DD}$ has positive partial transpose with respect the bipartition $\left\lfloor {\frac{N}{2}} \right\rfloor:\left(N-\left\lfloor {\frac{N}{2}} \right\rfloor\right)$ if, and only if,

\begin{equation}
M_{l}(\{p_k\}_{k=0}^{N(d-1)})= \left[
\begin{matrix}
p_{0+l}/\mathcal{N}_{0+l}&p_{1+l}/\mathcal{N}_{1+l}&  \dots& p_{n_l+l}/\mathcal{N}_{n_l+l} \\
p_{1+l}/\mathcal{N}_{1+l}&p_{2+l}/\mathcal{N}_{2+l}& \dots  &p_{n_l+1+l}/\mathcal{N}_{n_l+l+l}  \\
\vdots & \vdots &  &  \vdots \\
p_{n_l+ l}/\mathcal{N}_{n_l+l} & p_{n_l +1 + l}/\mathcal{N}_{n_l+l+1}  & \cdots & p_{2n_l + l}/\mathcal{N}_{2n_l+l}
\end{matrix}\right] \geq 0,\\
\end{equation}
where $l\in\{0,1\}$ and $n_{l=0} = \left\lfloor {\frac{N(d-1)}{2}} \right\rfloor, n_{l=1} = \left\lfloor {\frac{N(d-1)-1}{2}} \right\rfloor $. 
\end{lem}

This allows us to formulate the basic result regarding the qudit diagonal states of Eq.\eqref{dicke_diagonal} : 

\begin{thm}
\label{thm:symm4d} Let $\Lambda_\alpha(\rho_{DD})=\sigma_{DD}=
{\rm Tr}(\rho)\id_\mathcal{D} + \alpha \rho_{DD}$ be the family of maps acting on diagonal qudit Dicke states. Let $\alpha$ be such that the two matrices $M'_0$ and $M'_1$  be positive semidefinite, 
where  $M'_0=M_0(\{1+\alpha p_k \}_{k=0}^{N(d-1)})$, $M'_1=M_1(\{1+\alpha p_k \}_{k=0}^{N(d-1)})$.

Similarly, if $\sigma_{DD} \ge 0$, and $\Lambda_\alpha^{-1}(\sigma_{DD})=\rho_{DD}\ge 0$, then $\sigma_{DD}$ is separable. 
\end{thm}
\begin{proof}
The map transforms $p_k \to p_k'=1+\alpha p_k$. Note that $p_k'$ are not normalized, but this does not affect positivity condition.
Thus, the relevant Hankel matrices for $\sigma_{DD}$ have indeed the form presented in Lemma \ref{Hankeld}.
\end{proof}

Note that for qubits ($d=2$), $\rho_{DD}$ and $\rho_{DS}$ are equivalent, thus yielding Theorem \ref{thm:symm4} for $N$ even as a particular case.

For qutrits ($d=3$) and, more generally, qudits (unrestricted $d$) we can find the critical values of $\alpha$ by following these steps:
We begin by considering $\Pi_{DD} := \sum_{k=0}^{N(d-1)} \ket{D_{k}^N}\bra{D_k^N}$ and, for every $k\in \{0, \ldots, N(d-1)\}$ we compute the characteristic polynomial $Q_k(x;\alpha)$ of $(\Pi_{DD} + \alpha \ket{D_{k}^N}\bra{D_k^N})^{\Gamma}$, where $\Gamma$ denotes partial transposition with respect to some subsystems. By virtue of a convex roof argument, it suffices to consider projectors onto pure basis states $\ket{D_{k}^N}\bra{D_k^N}$. To analytically find a range of $\alpha$ we proceed by computing the polynomial greatest common divisor between $Q_{k}(x;\alpha)$ and $Q_{k}(x;0)$. Since $\alpha$ is treated as a parameter, the polynomial Euclidian algorithm using long division can be used, as $Q_k$ is a univariate polynomial. Computing $\gcd(Q_k(x;\alpha), Q_k(x;0))$ allows us to ignore all the eigenvalues of the partial transposition that do not depend on alpha without even having to compute them. All the eigenvalues that do not depend on $\alpha$ are non-negative if and only if $(\Pi_{DD})^\Gamma \geq 0$ on all possible partial transpositions. If this holds, we can therefore define $R_k(x;\alpha) := Q_k(x;\alpha)/\gcd(Q_k(x;\alpha), Q_k(x;0))$, which is a polynomial by construction. If all the roots of $R_k(x;0)$ are strictly positive, then by continuity of $\alpha$, we will have ${\alpha_{\max}} - {\alpha_{\min}} >0$. To find the candidates for ${\alpha_{\max}}$ and ${\alpha_{\min}}$ we can find all the roots of $R_k(0;\alpha)$ which can be expressed in terms of radicals for small $N$ and $d$. Note, however, that the projector $\Pi_{DD}$ violates the PPT criterion for $N=3$ and $d\geq 4$ and it is therefore entangled. In this case the range for $\alpha$ is $\emptyset$.

For qutrits and $N=3$, $R_k(x;\alpha)$ splits in cubic factors in $x$. Therefore, we can write their roots as radicals in function of $\alpha$, which we can denote $\lambda_i(\alpha)$. By solving for $\lambda_i(\alpha)= 0$ and keeping all the real solutions we can iterate on $k$ to find $\alpha_{\min}$ and $\alpha_{\max}$. In the case fof $N=d=3$, $\alpha_{\max}$ is given by $(7\sqrt{10}-22)/36 \approx 0.00377621$ whereas $\alpha_{\min}$ is the intermediate root of the polynomial $49\alpha^3 + 147\alpha^2 -315\alpha -1$, approximately $-0.00316992$.

\subsection{Breuer-Hall- inspired maps }

Let us start the subsection by presenting generalizations of Theorem~\ref{thm:barnum}
derived in a similar spirit as the
generalization of the reduction map by the Breuer-Hall map
\cite{Breuer.PRL.2006,Hall.JPA.2006}. Note that Breuer-Hall's construction does not work
in the two qubit case---Breuer-Hall's map identically vanishes for qubits. In our
case when we consider maps acting on the composite space of Alice
and Bob and Charlie, etc. this restriction will not apply. In the two-dimensional
space there exist (up to a phase) a single unitary,
$\sigma_2$, with the property that for every $|e\rangle$,
$\sigma_2|e^*\rangle=|e^{\perp}\rangle$, where $*$ indicates complex conjugation and $\perp$ perpendicular. This operator can be applied to all of the  $N$ qubits.
It is thus easy to construct analogous Breuer-Hall unitary
operators $V$, such that $V|f_1^*\rangle\otimes\ldots\otimes|f_k\rangle=|f_1^{\perp}\rangle\otimes\ldots\otimes|f_k^{\perp}\rangle$, being simple tensor products of $\sigma_2$.

Below we will denote, for brevity: 
\[\tilde\rho_A=\sigma^A_2\rho^{T_A}\sigma^A_2,\]
\[\tilde\rho_{AB}=\sigma^A_2\sigma^B_2\rho^{T_{AB}}\sigma^B_2\sigma^A_2,\] etc.

We can now prove the following result about a generalized $\Lambda_p$.

\begin{thm}
\label{thm:lewenstein1} Let $\Lambda^A_{\alpha,\beta}(\rho)= {\rm
Tr}(\rho)\id + \alpha \rho + \beta \tilde\rho_A$ be a family of
maps. Let $\alpha \ge {\rm max}[-1,\beta/2-1]$ and $\beta \ge
{\rm max}[-1,\alpha/2-1]$. Then $\rho\ge 0$ $\Rightarrow$
$\Lambda_{\alpha,\beta}(\rho)=:\sigma\in \Sigma_{N-1}$, i.e., $\sigma$ is
bi-separable in the partition $1^A:(N-1)^{BC\ldots}$
\end{thm}
\begin{proof}
The proof is very similar to that of Theorem \ref{thm:barnum} as above, and it follows \cite{LAC16}.
It is enough to prove it for pure
states, $\rho=|\Psi\rangle\langle \Psi|$. Without any loss of generality,
take $|\Psi\rangle =\lambda_0 |0\rangle_A\otimes
|0\rangle_{A.O.} + \lambda_1|1\rangle_A\otimes |1\rangle_{A.O.}$. It is enough to
check then positiveness and separability of
$\Lambda(|\Psi\rangle\langle \Psi|)$ on a $2\otimes 2$ space
spanned by $|0\rangle_{A, A.O.}$, $|1\rangle_{A, A.O.}$ in both Alice's and All Others' (A.O.)
spaces, where PPT provides then necessary and sufficient
condition. 
\end{proof}

The above theorem has an obvious corollary: 
\begin{cor}\label{cormany2}
Let $\Lambda^i_{\alpha_i,\beta_i}(\rho)= {\rm
Tr}(\rho)\id + \alpha_i\rho + \beta_i \tilde\rho_i$ be a family of
maps. Let $\Lambda(\rho)=\sum_{i=1}^Np_i\Lambda^i_{\alpha_i,\beta_i}(\rho)$, with $p_i\ge 0$, and $\sum_{i=1}^N p_i=1$. 
Let $\alpha_i \ge {\rm max}[-1,\beta_i/2-1]$ and $\beta_i \ge
{\rm max}[-1,\alpha_i/2-1]$. Then $\rho\ge 0$ $\Rightarrow$
$\Lambda(\rho)=:\sigma\in \Sigma_{N-1}$, i.e., $\sigma$ has entanglement depth $N-1$
(is a mixture of states that are bi-separable in the partition $1^i:(N-1)^{A.O.}$
\end{cor}
\begin{proof}
The proof is obvious: a mixture of bi-separable state in the partitions $1^i:(N-1)^{A.O.}$ has $(N-1)$-depth. The problem of using the above result for entanglement depth criteria lies in the difficulty of inverting the map $\Lambda(\rho)$. 
\end{proof}

The general Theorem~\ref{thm:general} then immediately implies the following $(N-1)$-depth result.
\begin{thm}
\label{thm:lewenstein2} {\bf (Sufficient $(N-1)$-depth criterion 2)}
Let $\rho =\Lambda_{\alpha,\beta}^{-1}(\sigma)\ge 0 $, i.e.,
\begin{equation}\label{MainResult:1}
\frac{1}{\alpha^2-\beta^2}\left[\alpha \sigma -\beta
\tilde\sigma_A - \frac{(\alpha-\beta)}{2^N+\alpha +\beta} {\rm
Tr}(\sigma) \id\right]\ge 0.
\end{equation}
Then $\sigma$ is bi-separable in the partition $1_A:(N-1)_{A.O.}$
\end{thm}

Noticeably, the contributions from $\alpha$ and $\beta$ compensate
each other, diminishing in effect off-diagonal parts of $\sigma$
in the Alice space.  The criterion derived above is
particularly strong at the boundary of the parameter region, i.e., for
$\beta=\alpha/2-1$ or $\alpha=\beta/2-1$.
\begin{cor}
\label{cor:lewenstein2} {\bf (Sufficient $(N-1)$-depth criterion 3)}
Let
\begin{equation}\label{MainResult:2a}
\left[\alpha\sigma -(\alpha/2-1) \tilde\sigma_A -
\frac{(1+\alpha/2)}{2N+ 3\alpha/2-1} {\rm Tr}(\sigma)
\id\right]\ge 0
\end{equation}
or
\begin{equation}\label{MainResult:2b}
\left[\beta\tilde\sigma_A -(\beta/2-1) \sigma -
\frac{(1+\beta/2)}{2N+ 3\beta/2-1} {\rm Tr}(\sigma) \id\right]\ge
0.
\end{equation}
Then $\sigma$ is bi-separable in the partition $1_A:(N-1)_{A.O.}$
\end{cor}

The above corollary has a particularly interesting limit, when both
$\alpha$ and $\beta$ tend to $\infty$.
\begin{cor}
\label{cor:lewenstein3} {\bf (Sufficient $(N-1)$-depth criterion 4)}
Let
\begin{equation}\label{MainResult:3a}
\sigma - \tilde\sigma_A/2 > 0
\end{equation}
or
\begin{equation}\label{MainResult:3b}
\tilde\sigma_A - \sigma/2 > 0.
\end{equation}
Then $\sigma$ is bi-separable in the partition $1_A:(N-1)_{A.O.}$
\end{cor}

Note that the latter inequalities are valid only in the asymptotic
sense and require strict positiveness. One should also note that in
the above results we used maps that
 explicitly involved partially transposed matrices. For this
reason, it is useful to remind the reader of some results of Ref.
\cite{Kraus+3.PRA.2000}.

\begin{thm}
\label{thm:Karnas} \cite{Kraus+3.PRA.2000} If $\sigma=\sigma^{T_A}\ge 0$
then
 $\sigma$ is bi-separable in the partition $1_A:(N-1)_{A.O.}$
\end{thm}

\begin{cor}
\label{cor:Karnas} \cite{Kraus+3.PRA.2000} Let $\sigma\ge 0$ be a state.
If $\sigma + \sigma^{T_A}$ is of full rank and $\|(\sigma +
\sigma^{T_A})^{-1}\|\ \| \sigma - \sigma^{T_A}\| \le 1$, then
$\sigma$ is bi-separable in the partition $1_A:(N-1)_{A.O.}$. This corollary implies that if $\sigma$ is of full rank and is very close to $\sigma^{T_A}$ then
 $\sigma$ is separable. Here we use the operator norm
$\|A\| := \max_{\|\Psi\rangle\|=1} \| A |\Psi\rangle\|$.
\end{cor}

Note that the present results  are  clearly independent since they
involve matrices $\tilde\sigma_A$ and related ones. Note that the corollary \ref{cor:Karnas} leads also to an obvious 
$(N-1)$-depth result, in analogy to the theorem  (\ref{thm:lewenstein2}), see below.
\begin{thm}
\label{thm:lewenstein2many} {\bf (Sufficient $(N-1)$-depth criterion 3)}
Let $\rho =\Lambda_{\alpha,\beta}^{-1}(\sigma)\ge 0 $, i.e.,
\begin{equation}\label{MainResult:1}
\frac{1}{\alpha^2-\beta^2}\left[\alpha \sigma -\beta
\tilde\sigma_A - \frac{(\alpha-\beta)}{2^N+\alpha +\beta} {\rm
Tr}(\sigma) \id\right]\ge 0.
\end{equation}
Then $\sigma$ has depth $(N-1)$.
\end{thm}

Let us now present the strongest theorem of this section, which
involves Breuer-Hall unitary operators on the both sides of Alice
and Bob:

\begin{thm}
\label{thm:lewenstein2} Let ${\bf p} =\{\alpha, \beta, \gamma,
\delta\}$ and $\Lambda_{\bf p}(\rho)= {\rm Tr}(\rho)\id + \alpha
\rho + \beta \tilde\rho_A + \gamma \tilde\rho_{A.O.}  + \delta
\tilde\rho$ be the family of maps; let there exist $0\le a, b\le
1$, $a+b\le 1$ such that the parameters fulfill the four conditions:
\begin{subequations}\label{Eq:cond.alpha-delta.a-b}
	\begin{align}
	\alpha &\ge \beta/2-a, \\
	\beta  &\ge \alpha/2-a, \\
	\gamma &\ge \delta/2-b, \\
    \delta &\ge \gamma/2-b,
	\end{align}
\end{subequations}
and $\alpha\ge -1$, $\beta\ge -1$, $\gamma\ge -1$, and $\delta\ge
-1$. Then $\rho\ge 0$ $\Rightarrow$ $\sigma = \Lambda_{\bf
p}(\rho)\in \Sigma$, i.e. $\sigma$ is separable.
\end{thm}

\begin{proof}
The proof is similar to the proof of Theorem~\ref{thm:lewenstein1},
but more complex and technical. Again it is enough to prove for
pure states, $\rho=|\Psi\rangle\langle \Psi|$ and, without loosing
generality, take $|\Psi\rangle =\lambda_0 |0\rangle\otimes
|0\rangle + \lambda_1|1\rangle\otimes |1\rangle$. But, now we have
to consider three cases: i) when the Breuer-Hall unitary $V$ acts
as $\sigma_2$ on Bob's subspace spanned by $|0\rangle$ and
$|1\rangle$; ii) when $V$ transforms $|0\rangle$ and $|1\rangle$
to $|2\rangle$ and $|3\rangle$; iii) when $V$ transforms
$|0\rangle$ and $|1\rangle$ to orthogonal vectors in a subspace
spanned by $|0\rangle$ and  $|2\rangle$, and $|1\rangle$ and
$|3\rangle$, respectively.

The conditions $\alpha\ge -1$, $\beta\ge -1$, $\gamma\ge -1$, and
$\delta\ge -1$ follow again from the analysis of the case when
$\rho$ is a projector on a product state, say
$|0\rangle\otimes|0\rangle$, so that $\tilde\rho_A$ is a projector
on $|1\rangle\otimes|0\rangle$, $\tilde\rho_B$ on
$|0\rangle\otimes|\tilde 1\rangle$, and $\tilde\rho$ on
$|1\rangle\otimes|\tilde 1\rangle$, where $|\tilde 1\rangle$ is a
vector orthogonal to $|0\rangle$. To derive the conditions
in Eq.~\eqref{Eq:cond.alpha-delta.a-b}, we split $\id= a\id + b\id + (1-a-b)\id$
and apply the result of Theorem \ref{thm:lewenstein1} to $a\id + \alpha\rho
+\beta\tilde \rho_A$, and $(1-a)\id + \gamma\tilde\rho_{A.O.}
+\delta\tilde\rho$. Clearly, this estimate of the region of
parameters, where $\Lambda_{\bf p}(\rho)$ is separable, is very
conservative, and probably can be improved significantly.
\end{proof}

The conditions above  are sufficient for separability and
correspond to the case iii), which probably is not the most
demanding, but we were not able to find weaker sufficient
conditions, i.e., the largest allowed regions of the parameters.
Case i) leads to the less demanding restrictions
that follow from the above conditions, when we set $a+b=1$.
 Finally, case
ii) is the least restrictive -- indeed it allows to exceed the
restriction $a+b\le 2$ and reach $a=b=1$.
Note that the proof of the Theorem \ref{thm:lewenstein1} is just
the same,  when we set $\gamma=\delta=0$, while  the proof of Theorem~\ref{thm:lewenstein2}, as discussed above, is more involved.

Of course, there exists a generalization of the above theorems to the case of mixtures: 
\begin{cor}\label{cormany2}
Let $\Lambda^i_{\alpha_i,\beta_i,\gamma_i, \delta_i}(\rho)= {\rm
Tr}(\rho)\id + \alpha_i\rho + \beta_i \tilde\rho_i+ \gamma_i \tilde \rho_{A.O.\ than\ i}+ \delta_i \tilde \rho$ be a family of
maps. Let $\Lambda(\rho)=\sum_{i=1}^Np_i\Lambda^i_{\alpha_i,\beta_i}(\rho)$, with $p_i\ge 0$, and $\sum_{i=1}^N p_i=1$. 
Let $\alpha_i$, $\beta_i$, $\gamma_i$ and $\delta_i$ fulfill the conditions of the theorem (\ref{thm:lewenstein2}). Then $\rho\ge 0$ $\Rightarrow$
$\Lambda(\rho)=:\sigma\in \Sigma_{N-1}$, i.e., $\sigma$ has entanglement depth $N-1$
(is a mixture of states that are bi-separable in the partition $1^i:(N-1)^{A.O.}$).
\end{cor}
\begin{proof}
The proof is obvious: a mixture of bi-separable state in the partitions $1^i:(N-1)^{A.O.}$ has $(N-1)$-depth. The problem of using the above result for entanglement depth criteria lies in the difficulty of inverting the map $\Lambda(\rho)$. Below we present explicit calculations for the case of 3 qubits.
\end{proof}

For simplicity, we consider a map 
\begin{cor}\label{cormany3}
Let 
\begin{eqnarray} \Lambda(\rho)&= & \sigma = {\rm Tr}(\rho)\id \\
&+&p_A[\alpha_A\rho + \beta_A \tilde\rho_A+ \gamma_A \tilde \rho_{BC}+ \delta_A \tilde \rho]  \nonumber\\
&+&p_B[\alpha_B\rho + \beta_B \tilde\rho_B+ \gamma_B \tilde \rho_{AC}+ \delta_B \tilde \rho] \nonumber\\
&+&p_C[\alpha_C\rho + \beta_C\tilde\rho_C+ \gamma_C \tilde \rho_{AB}+ \delta_C\tilde \rho] ,\nonumber                                    \end{eqnarray}
with $p_i\ge 0$, and $\sum_{i=A,B,C} p_i=1$. 
Applying different transpositions and $\sigma_2$s, we get:
\begin{eqnarray}  
\tilde\sigma_A = {\rm Tr}(\rho)\id \\
&+&p_A[\alpha_A\tilde\rho_A + \beta_A \rho+ \gamma_A \tilde \rho+ \delta_A \tilde \rho_{BC}]  \nonumber\\
&+&p_B[\alpha_{B}\tilde\rho_A + \beta_B \tilde\rho_{AB}+ \gamma_B \tilde \rho_{C}+ \delta_B \tilde \rho_{BC}] \nonumber\\
&+&p_C[\alpha_C\tilde\rho_A + \beta_C\tilde\rho_{AC}+ \gamma_C \tilde \rho_{B}+ \delta_C\tilde \rho_{BC}] ,\nonumber                                    \end{eqnarray}
and so on. In this way we obtain 8 linear inhomogeneous equations for $\rho$, $\tilde\rho_A$, $\ldots$, which can be easily solved to find $\rho=\Lambda^{-1}(\sigma)$. We use also the relation ${\rm Tr}(\sigma)= [8+\sum_{i=A,B,C}p_i (\alpha_i+\beta_i + \gamma_i + \delta_i)]{\rm Tr}(\rho)$. Obviously, we have then: If $\sigma\ge 0$ and $\rho=\Lambda^{-1}(\sigma)\ge 0$, then $\sigma$ does not exhibit any genuine 3-qubit entanglement.

Let $\alpha_i$, $\beta_i$, $\gamma_i$ and $\delta_i$ fulfill the conditions of the theorem (\ref{thm:lewenstein2}). Then $\rho\ge 0$ $\Rightarrow$
$\Lambda(\rho)=:\sigma\in \Sigma_{N-1}$, i.e., $\sigma$ has entanglement depth $N-1$
(is a mixture of states that are bi-separable in the partition $1^i:(N-1)^{A.O.}$).
\end{cor}

\section{\label{Other proposals} Sufficient criteria for compatibility}

In this section, we demonstrate how similar linear map methods can be used in the space of correlations to construct sufficient criteria for compatibility of a given set of correlations with certain models. In the first subsection, we consider $S'$ to be the set of statistics compatible with a global separable state. We finish by deriving sufficient criteria for compatibility with local-hidden variable models.

\subsection{Sufficient criteria for compatibility with a global entangled state}

Consider a system of $N$-qubits from which the first and second moments of the collective spin, $\mathbf{J} = \sum_{i=1}^N\boldsymbol{\sigma}^{(i)}/2$, are being inferred. Here we denote $\boldsymbol{\sigma}=(\sigma_x, \sigma_y, \sigma_z)$ the vector of Pauli matrices. Concretely, we focus on the unpolarized case $\langle \mathbf{J} \rangle = \mathbf{0}$ and assume knowledge of the second moments  $\{\langle J_a^2 \rangle\}_{a\in \{x,y,z\}}$ for three orthogonal directions $\{x,y,z\}$. Based on these expectation values, we want to verify compatibility with a global $N$-qubit separable state. Remarkably, there is a closed-form answer to this question, which we remind the reader in the following Lemma:

\begin{lem}\label{Toth}
\cite{Toth}
If the expectation values  $\{\langle J_a \rangle, \langle J_a^2 \rangle\}_{a\in \{x,y,z\}}$ fulfill the inequalities
\begin{equation}
    \label{Toth}
    \begin{array}{crl}
   \sum_{a\in\{x,y,z\}}  \langle J_a^2\rangle \leq N(N+2)/4 \\
      \sum_{a\in\{x,y,z\}}  ( \Delta J_a)^2 \geq N/2 \\
    \langle J_a^2 \rangle +  \langle J_b^2 \rangle -(N-1)(\Delta J_c)^2 \leq N/2\\
    (N-1)((\Delta J_a)^2 +(\Delta J_b)^2)-\langle J_c^2\rangle\geq N(N-2)/4,
    \end{array}
\end{equation}
where $(\Delta J_a)^2 = \langle J_a^2\rangle -\langle J_a\rangle^2$ is the variance and in the last two inequalities $\{a,b,c\}$ vary over all the permutations of $\{x,y,z\}$. Then, in the thermodynamic limit $N\rightarrow \infty$ there exists a global separable state $\rho_{\rm SEP}\in \mathcal{B}((\mathbb{C}^2)^{\otimes N})$ exhibiting $\{\langle J_a \rangle_{{\rm SEP}}, \langle J_a^2 \rangle_{{\rm SEP}}\}_{a\in \{x,y,z\}}$ arbitrary close to $\{\langle J_a \rangle ,\langle J_a^2 \rangle\}_{a\in \{x,y,z\}}$.

\end{lem}

Now, we apply the reduction map to the underlying state $\Lambda_\alpha(\rho) = \mathrm{Tr}(\rho)\id + \alpha \rho$. Instead of transforming the second moment $\langle J_a^2 \rangle$, we take equivalently $\langle J_a^2 \rangle - N/4=\langle \mathcal{C}_a \rangle$, where $\langle \mathcal{C}_a \rangle = \sum_{i\neq j= 1}^N\langle \sigma^{(i)}_a\sigma^{(j)}_a \rangle/4$ is the correlation function. The action of the map on the state rescales the correlation function $\langle \mathcal{C}_a \rangle\mapsto \alpha \langle \mathcal{C}_a \rangle :=\langle \mathcal{C}_a \rangle_\alpha$ equally for any orientation $a\in\{x,y,z\}$. 

In order to derive a sufficient condition for compatibility with a separable state, we need to find the range of $\alpha$ such that the modified correlations $\{\langle \mathcal{C}_a\rangle_\alpha \}_{a\in\{x,y,z\}}$ do not violate any inequality \eqref{Toth}. Provided that the state is unpolarized $\langle\mathbf{J}\rangle = 0$, $\{(\Delta J_a)^2 = \langle J_a^2\rangle = \langle \mathcal{C}_a \rangle + N/4\}_{a\in\{x,y,z\}}$, the parameter $\alpha$ appears linearly in the modified inequalities \eqref{Toth}. Consequently, for fixed correlations $\{\langle \mathcal{C}_a\rangle \}_{a\in\{x,y,z\}}$ the extreme value of $\alpha$ may be found via the linear program (LP): 
\begin{equation}
    \label{LP_sep}
    \begin{array}{crl}
    \max_{ \alpha}&\alpha&\\
    \mbox{s.t.}&-N\leq 4\alpha \sum_{a\in\{x,y,z\}}  \langle \mathcal{C}_a\rangle &\leq N(N-1) \\
    &  4\alpha (\langle \mathcal{C}_a\rangle+\langle \mathcal{C}_b\rangle-(N-1)\langle\mathcal{C}_c\rangle)&\leq N(N-1)\\
    &  4\alpha((N-1)(\langle\mathcal{C}_a\rangle+\langle \mathcal{C}_b\rangle)-\langle \mathcal{C}_c\rangle)&\geq N(1-N)\\
    \end{array}
\end{equation}

Note that in the same LP \eqref{LP_sep} one can consider minimization instead of maximization to find the extreme value in the opposite direction. The range of $\alpha$ is all we need to formulate the sufficient condition: 

\begin{thm}
\label{thm:Toth_sufficient} {\bf (Compatibility with a global separable state)}
Let $\rho\in \mathcal{B}((\mathbb{C}^2)^{\otimes N})$ be a quantum state and consider the set of expectation values $\{\langle J_a^2\rangle, \langle J_a\rangle = 0\}_{a\in\{x,y,z\}}$ descending from it. Let $\alpha^*$ be the solution of LP (\ref{LP_sep}).
If $\alpha^* > 1$, then $\sigma:=\Lambda_\beta^{-1}(\rho)/\mathrm{Tr}[\Lambda_\beta^{-1}(\rho)] =((2^N+\beta)\rho - \mathbbm{1})/\beta$ admits a separable description in the thermodynamic limit $N\rightarrow\infty$ from the same set of observables for $\beta \in [2^N/(\alpha^*-1),\infty)$
\end{thm}
\begin{proof}
Optimality of $\alpha$ implies that $\{\langle \mathcal{C}_a\rangle_{ \alpha^*} = \langle J_a^2 \rangle_{\alpha^*} - N/4 \}_{a\in\{x,y,z\}}$ do not violate any inequality \eqref{Toth} for all $\alpha \in [0, \alpha^*]$.  Consequently, according to Lemma \ref{Toth}, in the thermodynamic limit, it exists a separable state explaining the expectation values. Since the observables are fixed and $\{\langle \mathcal{C}_a\rangle\}_{a\in\{x,y,z\}}$ is generated from $\rho$, we just have to find the relation between $\alpha$ and $\beta$. For a given $\beta$, we have that $\{\langle \mathcal{C}_a\rangle_\alpha = \alpha\langle \mathcal{C}_a\rangle = \frac{2^N+\beta}\beta\langle \mathcal{C}_a \rangle \}_{a=\{x,y,z\}} $. Hence, $\beta = 2^N/(\alpha-1)$ yields the range $\beta \in [2^N/(\alpha^*-1),\infty)$.
\end{proof}

\subsection{Sufficient criteria for compatibility with a local-hidden-variable theory}

 The last set $S'$ we want to consider is that of the correlations that are compatible with a local hidden variable model (LHVM) \cite{Bell64}.
The set of LHVM correlations forms a polytope. Let us denote this polytope $\mathbbm{P}$. Membership of a set of correlators $\vec{P}(\vec{a}|\vec{x})$ in $\mathbbm{P}$ can be determined by the feasibility of the following LP:
\begin{equation}
    \label{eq:LP1}
    \begin{array}{crl}
    \min_{\vec{\lambda}}&0&\\
    \mbox{s.t.}&\sum_i \lambda_i \vec{v}_i& = \vec{P}(\vec{a}|\vec{x})\\
    &\sum_i \lambda_i&=1\\
    &\lambda_i &\geq 0,\\
    \end{array}
\end{equation}
where $\vec{v}_i$ enumerate all possible deterministic local strategies. If we take the no-signalling principle into consideration, in a Bell scenario with $n$ parties, $m$ measurement choices each with $d$ possible outcomes, a vertex $\vec{v}_i$ is given by $(1+m(d-1))^n-1$ coordinates, and $i$ ranges from $0$ to $d^{mn}-1$. The local deterministic strategies encoded in $i$ can be decoded by expressing the integer $i$ with digits $i_1, \ldots, i_n$ in base $d^m$, and then each $i_j$ as ${i_j}_1, \ldots, {i_j}_n$ as digits in base $d$. Then, ${i_j}_k$ indicates that party $j$ outputs the ${i_j}_k$-th outcome upon measuring the $k$-th observable, thus defining a deterministic local strategy.

Note that quantum correlations take the form $\vec{P}(\vec{a}|\vec{x}) = \mathrm{Tr}[\rho \Pi_{\vec{a}|\vec{x}}]$, where $\rho$ is some quantum state and $\Pi_{\vec{a}|\vec{x}}$ a POVM element. The image of the reduction map $\Lambda_\alpha$ on any $\rho$ shrinks the correlation vector $\vec{P}$ by a factor $\alpha$ and adds a small bias: $\mathrm{Tr}[\Lambda_\alpha(\rho) \Pi_{\vec{a}|\vec{x}}] = \mathrm{Tr}[(\mathrm{Tr}(\rho) \mathbbm{1}+ \alpha \rho)\Pi_{\vec{a}|\vec{x}}] = \mathrm{Tr}[\Pi_{\vec{a}|\vec{x}}] + \alpha \mathrm{Tr}[\rho \Pi_{\vec{a}|\vec{x}}]$. In the case of von Neumann, rank one projective measurements, $\mathrm{Tr}[\Pi_{\vec{a}|\vec{x}}] = 1$. If we take $\rho$ to be an $N-$partite $d$-level quantum state, then a proper normalization of $\Lambda_\alpha(\rho)$ maps $\vec{P}$ to $(\vec{1} + \alpha\vec{P})/(d^N+\alpha)$, where $\vec{1}$ is a vector of ones.

It is customary to take into consideration correlation functions instead of probabilities, to get rid of these inconvenient biases: In a Bell scenario with dichotomic measurements, one would normally label the measurement outcomes $\pm 1$ (and in more general scenarios, as $d$-th roots of unity \cite{ArnaultJPA}). For instance, for a bipartite scenario with dichotomic measurements labelled $\pm 1$, one would have $\langle A_x B_y\rangle := \sum_{a,b} (-1)^{a+b}P(ab|xy)$; \textit{i.e.,} the discrete Fourier transform of the correlators \cite{AugusiakNJP, SATWAP, MaxMaxMax}. In this case, the vector of correlators $\langle \vec{A}_{\vec{x}}^{(\vec{k})}\rangle$ maps to $\alpha \langle \vec{A}_{\vec{x}}^{(\vec{k})}\rangle$ since $\langle \vec{A}_{\vec{x}}^{(\vec{k})}\rangle_{\mathbbm{1}} = 0$ (we omit the correction for normalization of $\Lambda_\alpha(\rho)$ without loss of generality). Note that for binary outcomes the superindex $(\vec{k})$ is not necessary, and that whenever $k_i = 0$ this corresponds to a marginal correlator that does not involve the $i$-th party.

Hence, the condition for a set of correlators $\langle \vec{A}_{\vec{x}}^{(\vec{k})}\rangle$ to admit a LHVM interpretation reads
\begin{equation}
    \label{eq:LP2}
    \begin{array}{crl}
    \max_{\vec{\lambda}, \alpha}&\alpha&\\
    \mbox{s.t.}&\sum_i \lambda_i \vec{w_i}& = \alpha \langle \vec{A}_{\vec{x}}^{(\vec{k})}\rangle\\
    &\sum_i \lambda_i&=1\\
    &\lambda_i &\geq 0,\\
    \end{array}
\end{equation}
where $\vec{w}_i$ are the same $\vec{v}_i$ that have undergone the affine transformation given by the discrete Fourier transform. Naturally, in LP (\ref{eq:LP2}) one may as well consider $\min_{\vec{\lambda},\alpha} \alpha$ instead of the maximization to obtain the scaling factor in the opposite direction.

Along the lines of Theorem \ref{thm:Toth_sufficient}, we can similarly pose a sufficient condition for compatibility with a LHVM:

\begin{thm}
\label{thm:LHVM} {\bf (Compatibility with a local hidden variable model)}
Let $\rho$ be a quantum state and consider a set of traceless local observables $A_x^{{[i]}^{(k)}}$, yielding a vector of correlations
$\langle \vec{A}_{\vec{x}}^{(\vec{k})}\rangle_\rho$ on $\rho$. Let $\alpha^*$ be the solution of LP (\ref{eq:LP2}).
If $\alpha^* > 1$, then $\sigma:=\Lambda_\beta^{-1}(\rho)/\mathrm{Tr}[\Lambda_\beta^{-1}(\rho)] =((d^N+\beta)\rho - \mathbbm{1})/\beta$ admits an LHVM for the same set of observables for $\beta \in [d^N/(\alpha^*-1),\infty)$.
\end{thm}
\begin{proof}
Optimality of $\alpha^*$ implies that $\alpha \langle \vec{A}_{\vec{x}}^{(\vec{k})}\rangle_\rho \in \mathbbm{P}$ for all $\alpha \in [0, \alpha^*]$. Since the observables are fixed and $\langle \vec{A}_{\vec{x}}^{(\vec{k})}\rangle_\rho$ is generated from $\rho$, we just have to find the relation between $\alpha$ and $\beta$. For a given $\beta$, we have that $\alpha\langle \vec{A}_{\vec{x}}^{(\vec{k})}\rangle_{\rho} = \langle \vec{A}_{\vec{x}}^{(\vec{k})}\rangle_\sigma = \frac{d^N+\beta}{\beta} \langle \vec{A}_{\vec{x}}^{(\vec{k})}\rangle_\rho $. Hence, $\beta = d^N/(\alpha-1)$ yields the range $\beta \in [d^N/(\alpha^*-1),\infty)$.
\end{proof}

Since the characterization of $S'$ quickly becomes intractable in the multipartite case ($S'$ is a  polytope living in an $O(\exp(N))$-dimensional space determined by $O(\exp(N))$ vertices), here we may consider a relaxed version of the problem, by focusing on the projection of $S'$ onto the space of permutationally invariant one- and two-body correlators \cite{TuraScience, TuraAnnPhys}. In the latter case, $S'$ is embedded in a five-dimensional affine space and its vertices can be enumerated by the partitions of $N$, thus making the linear program in Eq. (\ref{eq:LP2}) efficient in practice.

In that case, if we consider $N$-partite permutationally invariant Bell inequalities with one- and two-body correlators for the two-input two-output scenario, one only needs to consider symmetric correlators ${\cal S}_x := \sum_i \langle A_x^{[i]}\rangle$ and ${\cal S}_{x,y} := \sum_{i\neq j} \langle A_{x,y}^{[i,j]}\rangle$, with $0 \leq x \leq y \leq 1$. The vertices $\vec{v}_i$ are such that ${\cal S}_{x} = p^{(1)}_x(c_i)$, where $p^{(1)}_x$ is a linear polynomial and $c_i \vdash N$ is a partition of $N$ in $d^m = 4$ elements, and ${\cal S}_{x,y} = p^{(2)}_{x,y}(c_i)$.
The analogous of the LP (\ref{eq:LP2}) for this case is
\begin{equation}
    \label{eq:LP3}
    \begin{array}{crl}
    \max_{\vec{\lambda}, \alpha}&\alpha&\\
    \mbox{s.t.}&\sum_i \lambda_i \vec{w_i}& = \alpha [{\cal S}_0,{\cal S}_1,{\cal S}_{00},{\cal S}_{01},{\cal S}_{11}]\\
    &\sum_i \lambda_i&=1\\
    &\lambda_i &\geq 0.\\
    \end{array}
\end{equation}
In this case it follows that:
\begin{cor}
\label{cor:2bodyPIBI} {\bf (Local Hidden Variable Model explainable through symmetric correlators)}
For the same conditions as in Theorem \ref{thm:LHVM}, let $\tilde\alpha$ be the optimal solution of LP (\ref{eq:LP3}). Then $\sigma$ admits a LHVM describable through symmetrized local response functions expressable as one- and two-body marginals, for $\beta \in [d^N/(\tilde\alpha-1),\infty)$.
\end{cor}

For very complex Bell scenarios, even the polynomial scaling of LP (\ref{eq:LP3}) may be too computationally intensive in practice. In this case, one can further relax the membership problem in the set $S'$ to that of a superset $S''$ via a symmetrized moment method based on outer approximations of convex hulls of semialgebraic sets, which admits a much more efficient description as a semidefinite program that does not depend on $N$ \cite{FadelTuraPRL}:
\begin{equation}
    \label{eq:SDP1}
    \begin{array}{crl}
    \max_{\vec{y}, \alpha}&\alpha&\\
    \mbox{s.t.}&y_i & = \alpha {(\vec{\cal S})}_i\\
    &y_0&=1\\
    &\sum_i y_i \Gamma_i &\geq 0,\\
    \end{array}
\end{equation}
where $\Gamma_i$ are constant matrices defined by $p_x^{(1)}$ and $p_{xy}^{(2)}$ \cite{FadelTuraPRL}.

\begin{cor}
\label{cor:ThetaBody} {\bf (Relaxation of Corollary \ref{cor:2bodyPIBI})}
Let $\vec{\cal S}$ be the vector of symmetrized correlators generated by $\rho$ and the same choice of observables at every site. Let us denote $\tilde\alpha$ the optimal solution to LP (\ref{eq:LP3}) and $\hat{\alpha}$ that of SDP (\ref{eq:SDP1}). Let $\Delta$ be a uniform (independent of $\vec{\cal S}$) upper bound to $\hat\alpha - \tilde \alpha$. If $\hat\alpha - \Delta > 1$ then $\sigma$ fulfills the conditions of Corollary \ref{cor:2bodyPIBI} for $\beta \in [d^N/(\hat \alpha - \Delta - 1), \infty)$.
\end{cor}


\section{\label{Conclusions and outlook} Conclusions and outlook}

We have presented in this paper several families of linear maps, that have a property that when applied to $N$-qubit states, result in states separable with respect a bipartition. When convertible, such maps allow deriving sufficient criteria for entanglement depth in $N$ qubit systems. In particular, in the present work we derive sufficient criteria for arbitrary entanglement depth tailored to states close to the normalized identity matrix and sufficient entanglement criteria for different classes of symmetric and symmetric-diagonal states. Finally, we propose how the application of linear maps can be used to detect states leading to expectation values compatible with separable states and local-hidden-variable theories. Such criteria have rather obvious possible applications and implications in quantum information science and quantum many body physics: they can namely be applied to determine whether and how much entanglement is needed for various specific protocols for quantum computation or other quantum tasks, similarly as it was discussed in Ref.~\cite {Braunstein+5.PRL.1999}.

\begin{acknowledgments}
M.L. and G.M.R. acknowledges support from: ERC AdG NOQIA; Ministerio de Ciencia y Innovation Agencia Estatal de Investigaciones (PGC2018-097027-B-I00/10.13039/501100011033, CEX2019-000910-S/10.13039/501100011033, Plan National FIDEUA PID2019-106901GB-I00, FPI, QUANTERA MAQS PCI2019-111828-2, QUANTERA DYNAMITE PCI2022-132919, Proyectos de I+D+I "Retos Colaboraci\'on" QUSPIN RTC2019-007196-7); Fundaci\'o Cellex; Fundaci\'o Mir-Puig; Generalitat de Catalunya (European Social Fund FEDER and CERCA program (AGAUR Grant No. 2017 SGR 134, QuantumCAT \ U16-011424, co-funded by ERDF Operational Program of Catalonia 2014-2020); Barcelona Supercomputing Center MareNostrum (FI-2022-1-0042); EU Horizon 2020 FET-OPEN OPTOlogic (Grant No 899794); 
EU Horizon Europe Program (Grant Agreement 101080086 - NeQST), 
National Science Centre, Poland (Symfonia Grant No. 2016/20/W/ST4/00314); European Union's Horizon 2020 research and innovation program under the 
Marie-Sk\l odowska-Curie grant agreement No 101029393 (STREDCH) and No 847648 ("La Caixa" Junior Leaders fellowships ID100010434: LCF/BQ/PI19/11690013, LCF/BQ/PI20/11760031, LCF/BQ/PR20/11770012, LCF/BQ/PR21/11840013). Views and opinions expressed in this work are, however, those of the author(s) only and do not necessarily reflect those of the European Union, European Climate, Infrastructure and Environment Executive Agency (CINEA), nor any other granting authority. Neither the European Union nor any granting authority can be held responsible for them. 
M.L., G.M.R. and A.S. acknowledge support from MCIN Recovery, Transformation and Resilience Plan with funding from European Union NextGenerationEU (PRTR C17.I1).
J.T. has received support from the European Union's Horizon Europe program through the ERC StG FINE-TEA-SQUAD (Grant No. 101040729).
A.S. acknowledges financial support from the Spanish Agencia Estatal de Investigaci\'on, Grant No. PID2019-107609GB-I00, the European Commission QuantERA grant ExTRaQT (Spanish MICINN project PCI2022-132965), and Catalan Government for the project
QuantumCAT 001-P-001644, co-financed by the European Regional Development Fund (FEDER).
\end{acknowledgments}

\end{document}